\documentclass[a4paper, 11pt]{article}
\usepackage[T1]{fontenc}
\usepackage[utf8]{inputenc}
\usepackage{latexsym}
\usepackage{amssymb,amsmath,amsthm}
\usepackage{graphicx}
\usepackage{wrapfig}

\newtheorem{tw}{Theory}[section]
\newtheorem{lm}[tw]{Lemma}
\newtheorem{obs}[tw]{Observation}

\newenvironment{pr}[1][Property]{\begin{trivlist}
\item[\hskip \labelsep {\bfseries #1}]}{\end{trivlist}}

\usepackage{indentfirst}

\begin{document}

\begin{titlepage}
\begin{center}
{\LARGE Uniwersytet Jagielloński} \\[5pt]
{\Large Wydział Matematyki i Informatyki} \\[3pt]
{\Large Zespół Katedr i Zakładów Informatyki Matematycznej}
\vspace{2.5cm}

{\Large Agnieszka Łupińska}\\

\vspace{3cm}
{\huge \bf A Parallel Algorithm to Test Chordality of Graphs}\\
\vspace{4cm}

\end{center}
\begin{flushright}
Promoter:\\
{Dr Maciej Ślusarek}\\
\end{flushright}
\vspace{1.5cm}
\begin{center}
{\Large Kraków 2013}
\end{center}
\end{titlepage}
\newpage
\tableofcontents
\newpage

\section{Introduction}

A graph $G$ is chordal if each cycle of size greater than 3 in $G$ has a~chord, 
that is an edge between two non-adjacent vertices on the cycle. 
We present a simple parallel algorithm to test chordality of graphs 
which is based on~the~parallel Lexicographical Breadth-First Search 
algorithm. In~total, the algorithm takes time $O(N)$ on $N$-threads machine and it 
performs work $O(N^2)$, where $N$ is the number of vertices in a graph. Our~implementation 
of the algorithm uses a GPU environment Nvidia CUDA~C. 
The~algorithm is implemented in CUDA 4.2 and it has been tested on Nvidia 
GeForce GTX 560 Ti of compute capability 2.1. At the end of the thesis 
we~present the~results achieved by our implementation and compare them 
with the results achieved by the sequential algorithm.

This thesis is organized as follows. Section 2 is an introduction to~the~parallel programming
using the GPU environment Nvidia CUDA C. Section 3 introduces the basic graph definitions used
throughout the paper. Then~it~provides an overview of the graph theory related to~the~LexBFS 
algorithm and chordal graphs. Section 4 introduces the LexBFS algorithm and its two most 
known implementations. Section 5 provides the sequential algorithm to test chordality of graphs 
and the analysis of its correctness and time complexity. In section 6 we 
present our parallel LexBFS algorithm and a parallel algorithm to test chordality of graphs. 
In section 7 we give the performance results of our parallel implementation compared to
the sequential algorithm. In~section 8 we discuss our results and the possible further work.
\newpage

\section{CUDA Programming}

CUDA (Compute Unified Device Architecture) is a general-purpose parallel computing architecture 
for Nvidia GPUs. We present the main features of CUDA C used in our implementation. For more details, 
we recommend NVIDIA CUDA C Programming Guide~[4] and CUDA C Best Practices Guide [5]. 

CUDA C extends the C/C++ programming model to the heterogeneous programming model which operates on 
the CPU called the host, and on the~GPU called the device. In CUDA, a kernel is a function
executed in~parallel by many threads on the device. A thread is a sequence of executions. The~threads 
are grouped into blocks which are grouped into a~grid. Each~thread has a unique identifier in a grid. 
It can be computed within a~kernel through a combination of the built-in variables: 
$threadIdx$, $blockIdx$ and $gridIdx$. 

All the threads may access data from the local, shared, constant, texture and global memory. To learn 
the texture memory and the constant memory see [4]. The local memory is a private memory of 
a thread. The shared memory is common to all threads within the same block and its lifetime is 
the same as the block. All theads have access to the same global memory.

The simple model of a program using the CUDA architecture is as follows: allocate and initialize
data on the host, allocate data on the device, transfer data from the host to the device, run the 
CUDA kernels on the device and transfer data from the device back to the host.

The CUDA architecture allows to synchronize executions of the threads in one block by using the 
$\_syncthreads$ function. It works as a lock: the~threads, which reach that point in the code, wait for 
other threads which have not done it yet. 

One of the methods to synchronize the threads between blocks, is to split the computitions in 
the synchronization points and to run each of that piece as a separate kernel. We use this method 
in our work.

\section{Background}

\subsection{Basic graph definitions}

We introduce the following terminology to be used throughout this thesis. Let $G=(V,E)$ be 
an undirected graph with the vertex set $V$ and the edge set $E$, where $E$ consists of unordered 
pairs of vertices in $V$. We denote the~size of $V$ by N and the size of $E$ by M. If $(u, v) \in E$ 
then we abbreviate it to $uv$. We use $N_x$ to denote the neighborhood of a vertex $x \in V$ 
excluding $x$. 

Let $\mathbb{N}$ be the set of the natural numbers and let $label_x$ be the label of~$x$, where 
$label_x$ is a string over the alphabet $\mathbb{N}$. We use $\circ$ to denote the~concatenation 
operator for labels.

A bijection $\pi = \{1, 2, \ldots, N\} \rightarrow V$ is called an \textit{ordering} of $G$.
Let~$\pi ^ {-1}$ denote the inverse of $\pi$ and thus $\pi^{-1}(v)$ is the index of $v$ 
in~the~ordering of~$G$. Let $\pi = v_1, \ldots, v_N$ be the ordering of $G$. We use $LN_{v_i}$ to
denote the~neighborhood of $v_i$ in the subgraph induced by $v_1, \ldots, v_{i-1}$.

We say that an ordering $\pi$ of $G$ is a BFS order if it is generated by the~well-known 
BFS algorithm (see for example [7]). We present this algorithm in~the~next chapter. 
Note that a graph can have many different BFS orderings.

A graph $G$ is \textit{chordal} if each cycle of size greater than 3 in $G$ has a \textit{chord}, 
that~is an edge between two non-adjacent vertices on the cycle. A vertex $x$ is \textit{simplicial} 
if $N_x$ induces a clique. An order $v_1, v_2, \ldots, v_k$ is a~\textit{perfect elimination order} 
if, for each $i$, $v_i$ is a simplicial vertex in the~graph induced by $v_1, v_2, \ldots, v_{i-1}$.

\subsection{Overview}

The LexBFS algorithm, in addition to the recognition of chordal graphs, has many other applications. 
The LexBFS algorithm is used as a part of~many graph algorithms such as recognizing interval graphs, 
or computing transitive orientation of comparability graphs, \text{co-comparability} graphs
and~interval graphs.

An \textit{orientation} of an undirected graph $G$ is a directed graph which is~created by assigning 
a direction to each edge. An orientation of edges is \textit{acyclic} if it does not contain a directed 
cycle. An orientation of edges is \textit{transitive} for all $x, y, z$, if $x \rightarrow y$ is an edge 
and $y \rightarrow z$ is an edge then $x \rightarrow z$ is also an edge. A \textit{comparability} graph 
is an undirected graph that has an~acyclic transitive orientation on edges. A \textit{co-comparability} 
graph is a~graph $G$ whose complement $\overline{G}$ is a comparability graph. An \textit{interval} 
graph $G$ is an undirected graph that is the~intersection graph of intervals on the real line, i.e.~
$G = (V, E)$, where $V = \{I_1, I_2, \ldots, I_n\}$, $\forall_i$ $I_i$ is an interval on~the~real line 
and $(I_i, I_j) \in E \Leftrightarrow I_i \cap I_j \ne \emptyset$. 

Gilmore and Hoffman [1] proved that a graph is an interval graph if~and~only if it is a chordal 
graph and a co-comparability graph (Figure 1). 

\begin{figure}[h]
    \begin{center}
        \includegraphics[width=6cm]{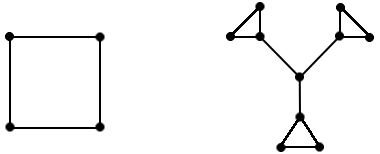}
        \caption{From left: $C_4$ is a co-comparability graph and it is not an interval graph.
        The next graph is a chordal graph and it is not an interval graph.}
    \end{center}
\end{figure}

\newpage

Habib, McConnell, Paul and Viennot [2] gave a $O(N+MlogN)$ algorithm for the~transitive orientation 
of a comparability graph and a~\text{$O(N+M)$} algorithm to recognize interval graphs. Both of them
use the LexBFS algorithm. It can be proved [2] that if $G$ is a co-comparability graph \text{and $\pi = 
v_1, v_2, \ldots, v_n$} is a LexBFS order of $G$ then there exists a transitive orientation 
of $\overline{G}$ such~that $v_n$ is a sink/source of the orientation. Moreover if $G$ is a 
comparability graph and $\pi = v_1, v_2, \ldots, v_n$ is a LexBFS order of $\overline{G}$ then there 
exists a~transitive orientation of $G$ such that $v_n$ is a sink/source of the orientation.

\section{Lexicographic Breadth-First Search}

The Lexicographic Breadth-First Search (LexBFS) algorithm was introduced by D. Rose, R. Tarjan 
and G. Lueker in 1976 for finding a perfect elimination order, if any exists. The LexBFS 
algorithm is a restriction of~the~widely used Breadth-First Search (BFS) algorithm in the 
following sense: each possible order of vertices produced by LexBFS is a BFS order. The 
difference between them is that the LexBFS algorithm additionally assigns labels to nodes 
and then in each step of the algorithm chooses a node, whose label is lexicographically 
the largest. 

\subsection{Characterization of BFS and LexBFS orderings}

We present and compare two characterizations of the vertex orderings that can by obtained 
by the BFS algorithm and the LexBFS algorithm.\\

Let $x$ be vertex of a graph $G$ and let $N_x$ denote its neighborhood. Let $Q$ be 
a FIFO queue. We present an equivalent version of Tarjan's BFS algorithm:\\

\begin{tabular}{l}
\hline
Breadth-Frist Search algorithm\\
\hline
\\
BFS()\\
\hspace{0.5cm}  for $x = 1$ to $n$ do $\pi^{-1}(x) = 0$\\
\hspace{0.5cm}  $Q = \emptyset$\\
\hspace{0.5cm}  for $i = 1$ to $n$ do\\
\hspace{1cm}        if $Q$.nonEmpty() then $x = Q$.dequeue()\\
\hspace{1cm}        else $x$ = any node s.t. $\pi^{-1}(x) = 0$\\
\hspace{1cm}        $\pi^{-1}(x) = i$, $\pi(i) = x$\\
\hspace{1cm}        for each $y \in N(x)$ s.t. $\pi^{-1}(y) = 0$ do\\
\hspace{1.5cm}          if $y \notin Q$ then $Q$.enqueue($y$)\\
\\
\hline
\end{tabular}
\\
\\
\begin{pr}[Property B.]
    If $a < b < c$, $ac \in E$ and $ab \notin E$ then exists $d < a$ such that $db \in E$. 
\end{pr}

\begin{lm}
    $\pi$ is a BFS order $\Leftrightarrow$ $\pi$ satisfies the B-property.
\end{lm}

\begin{proof}
    ($\Rightarrow$)
    Let $\pi$ be a BFS order and let $Q$ be a FIFO queue used by~the~algorithm. We assume 
    that the nodes of~the~graph~$G$ are~renumbered according to~$\pi$. 
   
    \begin{figure}[h]
        \begin{center}
            \includegraphics[width=3cm]{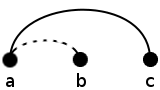}
        \end{center}
    \end{figure}

    During the algorithm $a$ was visited before $b$ which was visited before $c$.
    When the algorithm visits $a$ then it adds $c$ to $Q$, because of $ac \in G$, 
    and~it~does not add $b$ to $Q$, because of $ab \notin G$. Then $b$ can be first in $Q$ 
    before $c$ if and only if the algorithm had visited some $d$ before it visited $a$ 
    and $db \in G$ then the algorithm had added $b$ to $Q$ before it added $c$ to $Q$.
    In~the~same words, there exists $d < a$ such that $db \in G$.

    \begin{figure}[h]
        \begin{center}
            \includegraphics[width=4cm]{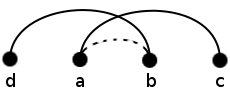}
        \end{center}
    \end{figure}

    ($\Leftarrow$)
    Let $\pi_0$ be an order satisfying a property $B$. We want to show that $\pi_0$ is
    a BFS order. Let $d$ be some vertex of graph $G$. When the BFS algorithm visits $d$ 
    it pushes on a queue all neighbors of $d$ which have not been visited yet. Then
    BFS pops the next vertex from a queue. Let $a$ be another vertex of a graph $G$. If the 
    BFS algorithm visits $d$ before $a$ then all not visited neighbors of $d$ are placed 
    in $\pi_0$ before all not visited neighbors of $a$. Hence it is sufficient to show 
    that if $d < a$ in order $\pi_0$, then all neighbors of $d$ which are placed on to the
    right of $d$ in $\pi_0$, lie before all neighbors of $a$ (which are to the right of $a$
    in $\pi_0$). But it is equivalent to property $B$. See the figure above.
    
\end{proof}

Let $x$ be a vertex of a graph $G$ and let $N_x$ be the neighborhood of $x$. Let~$label_x$ 
denote the label of $x$. Let $pQ$ be a priority queue of vertices with~priority on 
lexicographically the largest label. Consider the following algorithm:\\

\begin{tabular}{l}
\hline
Lexicographic Breadth-Frist Search algorithm\\
\hline
\\
LexBFS()\\
\hspace{0.5cm}  for $x$ = $1$ to $n$ do\\
\hspace{1cm}        $\pi^{-1}(x)$ = $0$\\
\hspace{1cm}        $label_x$ = $\emptyset$\\
\hspace{1cm}        $pQ$.enqueue($x$)\\
\hspace{0.5cm}  for $i$ = $1$ to $n$ do\\
\hspace{1cm}        $x$ = $pQ$.dequeue() //$label_x$ is lexicographically the largest\\
\hspace{1cm}        $\pi^{-1}(x)$ = $i$, $\pi(i)$ = $x$\\
\hspace{1cm}        for each $y \in N_x$ s.t. $\pi^{-1}(y) = 0$ do\\
\hspace{1.5cm}          $label_x$ = $label_x \circ (n-i)$\\
\\
\hline

\end{tabular}\\

\begin{pr}[Property LB.]
    If $a < b < c$, $ac \in E$ and $ab \notin E$ then exists $d < a$ such that $db \in E$ and 
    $dc \notin E$. 
\end{pr}

\begin{lm}
    $\pi$ is a LexBFS order $\Leftrightarrow$ $\pi$ satisfies LB-property.
\end{lm}

\begin{proof}
    ($\Rightarrow$) Let $\pi$ be the LexBFS order and let $pQ$ be a priority queue used 
    by~the~LexBFS algorithm. We assume that the nodes of the graph $G$ are~renumbered according to $\pi$. 
    
    \begin{figure}[h]
        \begin{center}
            \includegraphics[width=3cm]{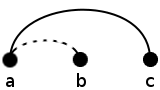}
        \end{center}
    \end{figure}

    During the algorithm $a$ was visited before $b$ which was visited before $c$.
    Let $b_1, b_2, \ldots \ $ be the label of $b$ and $c_1, c_2, \ldots \ $ be the label of $c$.
    Let $i$ be the~number of iteration and $N$ be the number of vertices in $G$.
    
    When the algorithm visits $a$, it concatenates the label of $c$ with \text{$N-i$}, as~$ac \in E$, 
    and it does not concatenate the label of $b$ with $N-i$, as \text{$ab \notin E$}. Since the~label 
    of $b$ is lexicographically larger than the label of $c$ then there~is an index $j_0$ in 
    the labels such that $\forall{j<j_0}:\ b_j = c_j$ and~$b_{j_0} > c_{j_0}$. The~index $j_0$ 
    is the first index at which the labels were updated in different iterations. The label of $b$ 
    was updated before the label of $c$ because $b_{j_0} > c_{j_0}$ and~in~each iteration the 
    number $N-i$ decreases as $i$ increases. The $j_0$ index exists if~and~only if the algorithm 
    visited some $d$ such that the algorithm concatenated the label of $b$ with $b_{j_0}$ and 
    it did not concatenate the label of $c$ with~$b_{j_0}$. It must be that $d < a$ because 
    otherwise we would have:

    \begin{enumerate}
        \item if $ad \in E$ then $a < d < c < b$, as the numbers $N-i$ decreased and~$ac \in E, ab\notin E$.
        \item if $ad \notin E$ then $a < c < d < b$, as $ac \in E$ and $ab \notin E$.
    \end{enumerate}
    Both cases are in contradiction to $a < b < c$. Therefore, there exists $d < a$ such that 
    $db \in E$ and $dc \notin E$. 

    \begin{figure}[h]
        \begin{center}
            \includegraphics[width=4cm]{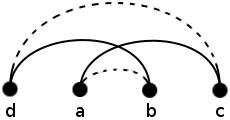}
        \end{center}
    \end{figure}
   
    ($\Leftarrow$) 
    Let $\pi_0$ be an order satisfying a property LB. We want to show that~$\pi_0$ is a LexBFS
    order. Let $d$ be a vertex of graph $G$. When the LexBFS algorithm visits $d$ then it updates 
    the labels of all not visited neighbors of~$d$. Let $a$ be another vertex of $G$. 
    If the LexBFS algorithm visits $d$ before $a$ then~all neighbors of $d$ which are not adjacent 
    to $a$ have labels greater than~the labels of all neighbors of $a$, because the numbers 
    added to~the~end of~labels decrease for successive vertices. So in order to prove the claim 
    we~must show that if $d < a$ in $\pi_0$ then all neighobrs of $d$, which are to~the~right of $d$
    in $\pi_0$ and they are not adjacent to $a$, lie before all neighbors of $a$, which~are to the right
    of $a$ in $\pi_0$. But again, it is equivalent to property LB. See the figure above.

\end{proof}

It is easy to see that the LB-property implies B-property so the LexBFS algorithm is 
a restriction of the BFS algorithm.

\subsection{Two implementations}

The first implementation of the LexBFS algorithm was proposed by D.J. Rose, R.E. Tarjan, G. S. 
Leuker in 1976 [3]. They use a double-linked list $L_k$ to store vertices of the same label $k$.
All lists $L_k$ are stored in the list $L$ in descending order given by labels $k$.
Additionally, each vertex $x$ has a pair of pointers the first of which is leading to the 
list $L_k$ containing $x$ and~the~second one is leading to the place of $x$ on the list $L_k$. 

At the beginning of the algorithm, all the vertices have the same label $\emptyset$ and they are 
on the list $L_{\emptyset}$. There are two operations: getting a vertex $x$ with the lexicographically
largest label and updating labels of all nodes adjacent to $x$. To perform the first 
operation the algorithm takes the first list $L_k$ from $L$ and then returns the first vertex from 
$L_k$. In the second operation, for each $y$ adjacent to $x$, the algorithm concatenates 
the label $k$ of the vertex $y$ with the number $(N-i)$. Then the algorithm removes $y$ from 
the list $L_k$ and inserts it to the list $L_{k \circ (N-i)}$, where $i$ is the iteration number. 
If~the~list $L_{k \circ (N-i)}$ does not existed in $L$ then the algorithm creates it.
This~implementation has the $O(N+M)$ time complexity. 

The second implemenation was proposed by Habib, McConnell, Paul and~Viennot in 2000 [2] and
it uses the partition refinement technique. Let~$V$ be a~doubly-linked list consisting of all 
vertices of $G$. Let $L$ be a doubly-linked list of classes of vertices. All vertices in a class
occupy consecutive elements in $V$ and the class is represented by a pair of pointers to the first 
and the last element in the class. Each~vertex $x$ has a pointer to the class containing~$x$.\\ 

\begin{tabular}{l}
    \hline
    The LexBFS algorithm using partition refinement\\
    \hline
    \\
    LexBFS()\\
    \hspace{0.5cm}  $L$ - a single-element list of a class containing all vertices\\
    \hspace{0.5cm}  for $i$ = $1$ to $n$ do\\
    \hspace{1cm}        $x$ - the first element of the first class on the list $L$\\
    \hspace{1cm}        remove $x$ from $L$\\
    \hspace{1cm}        $\pi^{-1}(x) = i,\ \pi(i) = x$\\
    \hspace{0.5cm}  //\textit{partition}\\
    \hspace{0.5cm}  for each class $C \in L$ do\\
    \hspace{1cm}        $C_x = C \cap N_x$\\
    \hspace{1cm}        $C_2 = C\ \textbackslash\ C_x$\\
    \hspace{1cm}        replace $C$ by $C_x$, $C_2$ in $L$\\
    \\
    \hline
\end{tabular}\\

During the partition, each $y \in N_x$ is removed from an old class $C$ and it is inserted
to some new class $C_x$. The \textit{partition} procedure can be implemented in $O(|N_x|)$ time.

\section{Chordal graphs}

Before we present the algorithm to test chordality of graphs we prove a~theorem introduced
by D.J. Rose, R.E. Tarjan, G. S. Leuker [3]. 


\begin{tw}[Rose, Tarjan, Leuker]
    A graph $G$ is chordal if and only if a LexBFS order of $G$ is a perfect elimination 
    order.
\end{tw}

\begin{proof}
    ($\Rightarrow$) Let $G$ be a chordal graph and let $\pi = v_1, \ldots, v_n$ be its LexBFS
    order. We assume that the nodes of $G$ are renumbered according to $\pi$. We~show that each 
    vertex $v_i$ is simplicial in the graph induced by $v_1, \ldots, v_{i-1}$. 

    Assume by contradiction that some $v_i$ is not simplicial. Then there exist $a, b \in \pi$ 
    such that $a < b < v_i$ both adjacent to $v_i$ and not adjacent themselves.

    \begin{figure}[h]
        \begin{center}
            \includegraphics[width=3.1cm]{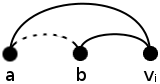}
        \end{center}
    \end{figure}

    Because $\pi$ satisfies LB-property then there exists some $c \in \pi$ \text{such that~$c < a$} and $c$ 
    is adjacent to $b$ and it is not adjacent to $v_i$. Note that $ca \notin E$ because $G$ is 
    chordal and otherwise we would have a chordless cycle $(c, a, v_i, b)$.

    \begin{figure}[h]
        \begin{center}
            \includegraphics[width=4.1cm]{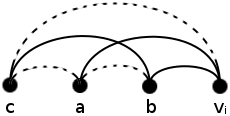}
        \end{center}
    \end{figure}

    Now we have got $c < a < b < v_i$, $ca \notin E$ and $cb \in E$. Again, we use the LB-property 
    in respect to $c, a, b$ and obtain $d \in \pi$: $d < c$, $da \in E$, $db \notin E$. Moreover $d$ 
    is not adjacent to $c$ because of the cycle $(d, c, b, v_i, a)$ and chordality of $G$.

    \begin{figure}[h]
        \begin{center}
            \includegraphics[width=5.1cm]{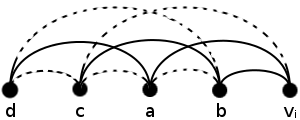}
        \end{center}
    \end{figure}

    Next time we can apply the LB-property to the vertices $d, c, a$. This step can be repeated 
    infinitely, thus contradicting the general assuption that $G$ is finite.

    So we have proved that each vertex $v_i$ is simplicial in the graph induced by $v_1, \ldots, 
    v_{i-1}$, that is the order $v_1, \ldots, v_n$ is the perfect elimination order.

    ($\Leftarrow$) It suffices to prove that if graph $G$ has any perfect elimination order
    then it is chordal.
    
    Let $\pi = v_1, \ldots, v_n$ be a perfect elimination order of $G$. Assume that in $G$ 
    there is a cycle $C$ of length $\ge 4$ and let $v_i \in C$ be the vertex of the greatest 
    index in $\pi$. Let $a$ and $b$ be vertices adjacent to $v_i$ in the cycle $C$. 
    Because $a$ and $b$ are on the left from $v_i$ in $\pi$ then they are adjacent as $\pi$ 
    is a perfect elimination order. Therefore $C$ has the chord $ab$, which finishes the proof.

\end{proof}

\subsection{Maximum Cardinality Search}

In 1984 Robert E. Tarjan and M. Yannakakis introduced in [6] the Maximum Cardinality Search
algorithm (MCS) as an alternative method for~finding a perfect elimination ordering of chordal 
graphs. The Maximum Cardinality Search instead of strings uses natural numbers as labels for 
vertices. In each iteration, the algorithm chooses a new vertex of the largest label, that is 
the vertex whose neighborhood in the graph induced by the nodes chosen so far is the largest 
among all vertices have not chosen yet. The MCS algorithm has a $O(N+M)$ time implementation [6].

Let $G$ be a graph, $pQ$ be a priority queue. For each $x \in V(G)$ $label_x \in \mathbb{N}$.
We present Tarjan and Yannakakis's MCS algorithm:\\

    \begin{tabular}{l}
        \hline
        Maximum Cardinality Search algorithm\\
        \hline
        \\
        MCS()\\
        \hspace{0.5cm}  for $x$ = $1$ to $n$ do\\
        \hspace{1cm}        $\pi^{-1}(x)$ = $0$\\
        \hspace{1cm}        $label_x = 0$\\
        \hspace{1cm}        $pQ$.enqueue($x$)\\
        \hspace{0.5cm}  for $i$ = $1$ to $n$ do\\
        \hspace{1cm}        $x$ = $pQ$.dequeue()\\
        \hspace{1cm}        $\pi^{-1}(x)$ = $i$, $\pi(i)$ = $x$\\
        \hspace{1cm}        for each $y \in N_x$ such that $\pi^{-1}(y) = 0$ do\\
        \hspace{1.5cm}          $label_x$ = $label_x + 1$\\
        \\
        \hline
        
    \end{tabular}\\

    Robert E. Tarjan and M. Yannakakis proved the following theorem [6].

\begin{tw}
    $G$ is a chordal graph if and only if $MCS$-order of $G$ is a perfect elimination order.
\end{tw}

\subsection{Algorithm to test chordality}

Theorem 5.1. gives us the following tool to test if a given graph is chordal. First we run 
the LexBFS algorithm to produce a LexBFS order. Next we~check if the LexBFS order is 
the perfect elimination order. 

Let $\pi = v_1, \ldots, v_n$ be an order returned by the LexBFS algorithm. For~each~$v_i$ 
let $LN_{v_i} \subset N_{v_i}$ be vertices adjacent to $v_i$ on the left from $v_i$ in $\pi$
and let~$p_{v_i} \in LN_{v_i}$ be the right most vertex in $LN_{v_i}$. 

The algorithm presented below tests if $\pi$ is a perfect elimination order. This is performed
by checking if for each $v_i$ it holds that $LN_{v_i}~-~\{p_{v_i}\}~\subset~LN_{p_{v_i}}$.
The correctness of such the approach is proved later.

To test if $\pi$ is a perfect elimiantion order we only need to check if~for~each~$v_i$
is $LN_{v_i} - \{p_{v_i}\} \subset LN_{p_{v_i}}$.\\ 

\begin{tabular}{l}
\hline
Test if a LexBFS order is a perfect elimination order\\
\hline
\\
chordalityTest()\\
\hspace{0.5cm}  for $x$ = $1$ to $n$ do $p_x = 0$\\
\hspace{0.5cm}  for each $y \in N_x$ that $\pi^{-1}(y) < \pi^{-1}(x)$ do\\
\hspace{1cm}        $LN_x$.add($y$)\\
\hspace{1cm}        if $\pi^{-1}(y) > \pi^{-1}(p_x)$ then $p_x = y$\\ 
\hspace{0.5cm}  for $x$ = $1$ to $n$ do\\
\hspace{1cm}        for each $y \in N_x$ do\\
\hspace{1.5cm}          $visited_y = 1$\\
\hspace{1cm}        for each $y \in N_x$ do\\
\hspace{1.5cm}          if $p_y = x$ then\\
\hspace{2cm}                // check if $LN_y - \{x\} \subset LN_x$\\
\hspace{2cm}                for each $z \in LN_y$ such that $z \ne x$ do\\
\hspace{2.5cm}                  if $visited_z \ne 1$ then\\
\hspace{3cm}                        return false\\
\hspace{1cm}        for each $y \in N_x$ do\\
\hspace{1.5cm}          $visited_y = 0$\\
\hspace{0.5cm}  return true\\
\\
\hline
\end{tabular}\\

\subsection{Correctness and complexity of algorithm}

We prove that the algorithm for testing if a LexBFS order is a perfect elimination order 
is correct.

Let $\pi$ be a perfect elimination order. We show that then the algorithm returns $true$. 
Let $v$ be some node of $G$. There are two cases:

\begin{enumerate}
    \item $p_v = 0$. 
        Then $LN_v$ is empty and the algorithm does not return false. 
    \item $p_v = u$, for some $u$.
        Then algorithm checks if $LN_v - \{u\} \subset LN_u$. 
    
        \begin{figure}[h!]
            \begin{center}
                \includegraphics[width=9cm]{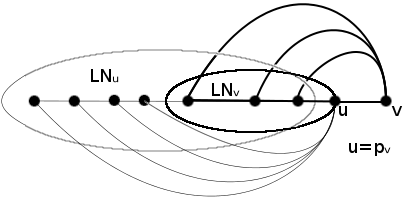}
            \end{center}
        \end{figure}

        Because all nodes of set $LN_v - \{u\}$ are on the left of $u$ then~they~are 
        candidates to members of $LN_u$. It remains to show that they are adjacent to~$u$. 
        Because $u = p_v$ then $u \in LN_v$ and $LN_v$ is a clique (because $\pi$ is a perfect 
        elimination order) therefore all vertices of $LN_v-\{u\}$ are adjacent to $u$ and 
        $LN_v - \{u\} \subset LN_u$. So the algorithm never return false, hence it returns
        true in the end.
\end{enumerate}

Let $\pi$ be not a perfect elimination order. We show that then the algorithm returns $false$.
We assume that vertices of $G$ are renumbered according to $\pi$. As $\pi$ is not a perfect 
elimination order, there exists some node in $\pi$ that~its left neighborhood does not induce 
a clique. Let $v$ be the first such~node in $\pi$ and let $p_v = u$ for some $u$. Because $u < v$ 
in $\pi$ then $LN_u$ is~a~clique (as $v$ is the~first vertex in $\pi$ for which $LN_v$ is not 
a clique). Therefore $LN_v - \{u\} \not\subset LN_u$ and the algorithm returns $false$.\\

The time complexity of the algorithm is determined by the nested loops. Let $v$ be some node of $G$. 
Note that the size of $LN_v$ is $O(|N_v|)$. Let us see how many times the algorithm scans $N_v$. 

\begin{enumerate}
    \item marking $visited$ array
    \item looking for $p_v$ 
    \item unmarking $visited$ array
    \item for $u$ such that $p_u = v$ and for each $v$ there is at most one such vertex~$u$
\end{enumerate}

It means that each list $N_v$ is read at most four times which gives $O(M)$ time for the whole 
graph. Summing up with the~time of producing the~LexBFS order, the test of chordality takes 
$O(N + M)$ time.

\section{Parallel algorithm}

Testing chordality of graphs has two steps: finding a LexBFS order and~checking if the LexBFS 
order is a perfect elimination order. To parallelize the chordality test we need to parallelize 
each of these steps separately. 

\subsection{Parallel LexBFS}

In our approach to the parallel version of LexBFS, the main loop of algorithm runs on the CPU and 
during each iteration $i$ two task are performed on the GPU. The first one is choosing the vertex 
$v$ with the lexicographically largest label and the second one is concatenating labels of vertices 
adjacent to $v$ with $N-i$. Both jobs are performed by N threads assigned to N vertices in a graph.

\begin{figure}[h]
    \begin{center}
        \includegraphics[width=7.5cm]{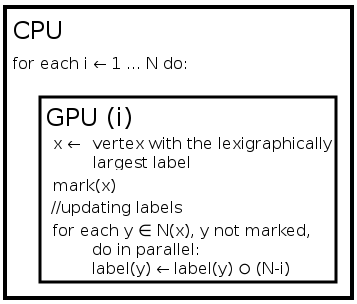}
        \caption{LexBFS algorithm}
    \end{center}
\end{figure}

We use the following data structures. The graph $G$ is stored in an adjacency matrix $Adj$. 
A linked list $L$ is a list of sets $L_k$. Each set $L_k$ includes all vertices whose labels 
are equal $k$. We identify the label of a set with~the~label of nodes in that set. These sets
form a partition of the vertex set, defined by means of labels. $L$ is sorted lexicographically 
ascending. As $L$ is a linked list, each set of $L$ has a pointer $next$ leading to the next 
set in the order, $next = NULL$ in the last set in $L$. The $order$ array stores the order 
of~nodes computed by the algorithm. We say that a node is $active$ if it has not processed yet.
(see Figure 3.)

\begin{figure}[h!]
    \begin{center}
        \includegraphics[width=6.5cm]{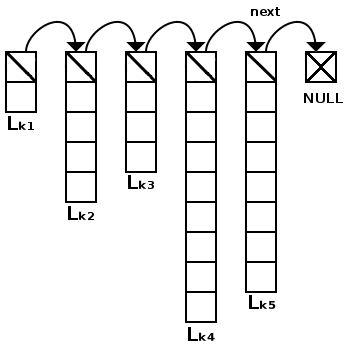}
        \caption{List $L$ including sets: $L_{k1}$, $L_{k2}$, $L_{k3}$, $L_{k4}$, $L_{k5}$. 
            The last set on~the~list is $L_{k5}$.}
    \end{center}
\end{figure}

At the begining of algorithm all nodes of $G$ are active and they have the~same label $\emptyset$. 
It means that the list $L$ has only one set $L_{\emptyset}$ consisting of~all nodes of $G$ and 
$next$ leads to $NULL$.

In a sequential version of algorithm, to find the vertex with the lexicographically largest label, 
the algorithm returns any vertex belonging to~the~last set on $L$. As the last set on $L$ is 
characterized by the pointer next equal $NULL$, then this procedure can be performed in parallel 
by $N$ threads as~follows. 

Let $x$ be the vertex assigned to the thread~$th_x$. Let $L_x$ be the set 
including $x$. Let $current$ be a global variable shared by all threads.
For~each~thread $th_x$ in parallel do:\\

\hspace{0.75cm} if $L_x.next = NULL$ then $current \leftarrow x$\\

After this procedure the $current$ variable stores a vertex whose label is~lexicographically 
the largest. Note that if there is more than one such~vertex then we cannot predict which one 
will be stored in $current$.

Let $i$ be the iteration of the main loop in which $current$ has the lexicographically largest label.
Let $y$ be some neighbor of $current$, $l_y$ be the label of $y$ and let $y$ be in the set $L_y$.

The update operation concatenates label $l_y$ in back with number $N-i$. Note that the number $N-i$ 
has not appeared in any label so far and it is the smallest among all numbers occuring in the labels. 

Next, the algorithm removes $y$ from $L_y$ and inserts it to the new set containing the nodes with
the label $l_y \circ (N-i)$. If the new set has not existed yet then the algorithm creates it.

Let us look closely at the operation of creating a new set. Let $A$ and~$B$ be~two sets of nodes on 
the list $L$ containing nodes labeled $l_A$ and $l_B$ respectively. Assume that $l_A < l_B$, i.e. 
$A$ comes before $B$ in $L$.

\begin{lm} If $j$ is a number of iteration during which the new set containing nodes with the label 
    $l_A \circ (N-j)$ is created then $l_A < (l_A \circ (N-j)) < l_B$ in $L$.
\end{lm}

\begin{figure}[h]
    \begin{center}
        \includegraphics[width=11cm]{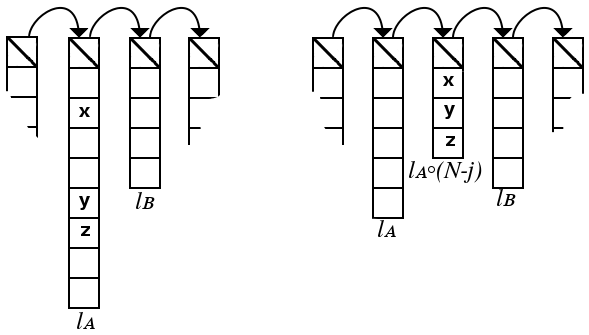}
        \caption{The list $L$ before and after moving $x$, $y$, $z$ to the new set. 
        Vertices $x, y, z$ are adjacent to $current$ vertex.}
    \end{center}
\end{figure}

\begin{proof} Each label is a string of numbers. Let $l_A = \{a_1, a_2, a_3, \ldots, a_{|l_A|} \}$, 
    $l_A \circ (N-j) = \{a_1, a_2, a_3, \ldots, a_{|l_A|}, N-j\}$ and 
    $l_B = \{b_1, b_2, b_3, \ldots, b_{|l_B|} \}$.
    Let~$k$ be the smallest index such that $\forall i<k:\ a_i = b_i$ and $a_k \ne b_k$. There~are two cases:
    \begin{enumerate}
        \item $k \le |l_A|$ and $k \le |l_B|$\\
            Then $a_k < b_k$ and this two numbers determine $l_A \circ (N-j) < l_B$. 
        \item $|l_A| < |l_B|$ and $k = |l_A| + 1$\\
            Then $l_A$ is a prefix of $l_B$ and concatenation also gives $l_A \circ (N-j) < l_B$ 
            because $(N-j)$ is less than all numbers in $l_B$ so in particular smaller than $b_k$.
    \end{enumerate}
    Note that always $l_B < l_B \circ (N-j)$ because of $|l_B| < |l_B \circ (N-j)|$.
\end{proof}

The lemma gives us the following observation:

\begin{obs}
When a new set is created for vertices from a given set $S$ then it should be inserted between $S$
and its successor on the list $L$. 
\end{obs}

Based on observation 6.1., for each new set we can determine its place in~the~list without 
any additional list traversal or label comparisons. 

How many new sets are created during one iteration? The answers is: at most one for each old one.
Indeed if $y$ and $z$ are neighbors of $current$ vertex and belonging to some set $S$ then their
labels are equal both before and after concatenating them with $N-i$.

Since in our algorithm updating labels is performed in parallel, it could happen that for some 
new label several threads would simultaneously create several new sets and then insert them to 
the list. In order to avoid such~a~mistake we use synchronization between performed instructions.
\newpage
Let $i$ be the number of iteration and let $current$ be a vertex with lexicographically the largest 
label chosen during $i$ iteration. Let $x$ be the vertex assigned to the thread $th_x$. Let $x$ 
belong to set $L_x$. Let $oldNext_x$ and~$newNext_x$ be private variables of the thread $th_x$. 
For~all threads $th_x$ in~parallel do:\\

\hspace{0.75cm} 1. if $x$ is not active or $x$ is not adjacent to $current$ then stop

\hspace{0.75cm} 2. $oldNext_x \leftarrow L_x.next$

\hspace{0.75cm} 3. create a new set $newNext_x$

\hspace{0.75cm} 4. synchronization: wait for other threads

\hspace{0.75cm} 5. set pointers: 

\hspace{1.4cm}      $L_x.next = newNext_x$

\hspace{1.4cm}      $newNext_x.next = oldNext_x$

\hspace{0.75cm} 6. synchronization: wait for other threads

\hspace{0.75cm} 7. insert $x$ to $L_x.next$\\

Each thread creates its new set $newNext_x$ and inserts it to the order. After that for each 
node $x$, $newNext_x.next = oldNext_x$ but only for one of~them we will have $L_x.next = newNext_x$. 
See Figure 5.

\begin{figure}[h]
    \begin{center}
        \includegraphics[width=8cm]{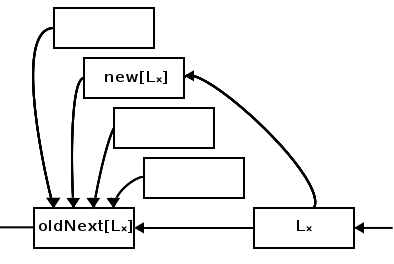}
        \caption{Insert a new set $newNext_x$ into the list $L$}
    \end{center}
\end{figure}

Note that we cannot predict which $newNext_x$ will be in $L_x.next$ so synchronization is performed 
to all threads read the same $L_x.next$ inserted after~$L_x$ in the list. The new sets of other 
threads are forgotten.

Now let us look at removing vertices from the sets. After this operation some sets can be empty. 
To get a vertex of the lexicographically largest label in the list, we take the last set on the 
list and this set cannot be empty. Therefore after each update operation, when removing vertices 
from sets is~performed, we remove all empty sets from the list.\\
\newpage
Before we show the procedure of removing the empty sets, consider the~following lemma. 

\begin{lm}
    If, after the update operation, the set $L_k$ is empty then its successor on the list, if
    exists, is nonempty.
\end{lm}

\begin{figure}[h]
    \begin{center}
        \includegraphics[width=9cm]{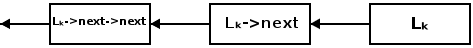}
    \end{center}
\end{figure}

\begin{proof}
    Let $L_k$ be not empty set before the update operation. If after the~update operation 
    $L_k$ is empty then from Observation 6.2. we know that the~set $L_k.next$ includes all 
    vertices which were in set $L_k$ before the update. 
\end{proof}

The observation that for each empty set, its predecessor and successor are~nonempty means that 
all empty sets can be removed in parallel in one time. The removal of one set requires changing 
only two links of two adjacent non-empty sets. It is correct because it does not require any 
additional traversing through the list and it is independent from removing other empty sets.

To find out which sets are empty, we use an additional $counter$ array. For each set in the list,
$counter$ stores $1$ if the set includes at least one node or $0$ otherwise. At~the beginning 
all slots of $counter$ are $0$. 

Let $x$ be the vertex assigned to the thread $th_x$. Let $x$ be in set $L_x$. Let $i$ be the number 
of iteration. Let $current$ be the vertex with lexicographically the largest label during $i$ 
iteration. Let $oldNext_x$ and $newNext_x$ be private variables of the thread $th_x$. 
For all threads $th_x$ in parallel do:\\

\hspace{0.75cm} 1. if $x$ is not active then stop

\hspace{0.75cm} 2. $counter[L_x] \leftarrow 1$

\hspace{0.75cm} 3. $oldNext_x \leftarrow L_x.next$

\hspace{0.75cm} 4. synchronizetion: wait for other threads

\hspace{0.75cm} 5. if $counter[oldNext_x] = 1$ then stop

\hspace{0.75cm} 6. set $L_x.next \leftarrow oldNext_x.next$\\

We use synchronization between instructions to make sure that counting vertices and removing 
empty sets are correct. Otherwise some sets could be~removed despite they are not empty and 
pointers could be changed improperly.\\
\newpage
In our implementation, getting the vertex of lexicographically the largest label and updating 
the labels of all adjacent vertices are performed in four stages. Each stage run on the GPU 
and they are synchronized - the~new one does not start until the last one has not finished. 
Because we use the~CUDA language to implement the algorithm we use the term \textit{kernel} 
instead of \textit{stage}. Each kernel is executed by $N$ threads, one for each vertex.

At the beginning of algorithm all vertices have lexicographically the same label then we set 
$current$ to random vertex $1$. Next, in each iteration of $for$ loop, one vertex is choosen 
to the LexBFS order. 

The first kernel adds $current$ vertex to the LexBFS order and marks it as~non-active. Next,
the first part of the labels updating is performed: setting counters of sets on $0$ and 
saving the $next$ pointers of each set.

In the second kernel, the new sets are inserted to the list $L$.

In the third kernel, each vertex that is active and adjacent to $current$ is moved to the next 
set. Next, the counting is performed: for each active vertex, the counter of its set is made 
equal $1$. At the end, the first part of~deleting the empty set is performed: saving the $next$ 
pointers of each set.

The last kernel deletes all empty sets and chooses the new $current$ vertex.

At the end of the LexBFS algorithm, the $order$ array stores the LexBFS order.\\

\begin{tabular}{l}
\hline
Parallel Lexicographic Breadth-Frist Search algorithm\\
\hline
\\
parallelLexBFS()\\
\hspace{0.5cm} $current \leftarrow 1$\\
\hspace{0.5cm} for $time \leftarrow 1$ to $n$ do\\
\hspace{0.5cm}\hspace{0.5cm} kernel1()\\
\hspace{0.5cm}\hspace{0.5cm} kernel2()\\
\hspace{0.5cm}\hspace{0.5cm} kernel3()\\
\hspace{0.5cm}\hspace{0.5cm} kernel4()\\
\\
kernel1()\\
\hspace{0.5cm} if $x$ is active then\\
\hspace{1cm}        $oldNext_x \leftarrow L_x.next$\\
\hspace{1cm}        $counter[L_x] \leftarrow 0$\\
\hspace{1cm}        create $newNext_x$\\
\hspace{0.5cm} if $x$ is $current$ then\\
\hspace{1cm}        $order[time] \leftarrow x$\\
\hspace{1cm}        mark $x$ as non-active\\
\\
kernel2()\\
\hspace{0.5cm} if $x$ is active and $x$ is adjacent to $current$ then\\
\hspace{1cm}        \textit{//inserting new sets to list}\\
\hspace{1cm}        $L_x.next \leftarrow newNext_x$\\
\hspace{1cm}        $newNext_x.next \leftarrow oldNext_x$\\
\\
kernel3()\\
\hspace{0.5cm} if $x$ is active and $x$ is adjacent to $current$ then\\
\hspace{1cm}        \textit{//moving to new sets}\\ 
\hspace{1cm}        move $x$ to $L_x.next$\\
\hspace{0.5cm} if $x$ is active then\\
\hspace{1cm}        $counter[L_x] \leftarrow 1$\\
\hspace{1cm}        $oldNext_x \leftarrow L_x.next$\\
\\
kernel4()\\
\hspace{0.5cm}  if $x$ is active then\\
\hspace{1cm}        if $counter[L_x] = 0$ then\\
\hspace{1.5cm}          \textit{//deleting empty sets}\\
\hspace{1.5cm}          $L_x.next \leftarrow oldNext_x.next$\\
\hspace{1cm}        if $L_x.next = NULL$ then\\
\hspace{1.5cm}          \textit{//updating current}\\
\hspace{1.5cm}          $current \leftarrow x$\\
\\

\hline

\end{tabular}\\

\newpage

\subsection{Parallel test for perfect elimination order}

Now we are given the LexBFS order $\pi$. The second step of testing chordality of graphs 
is checking if the LexBFS order is the perfect elimination order, that is, if for each 
vertex $x$, the neighborhood of $x$ on the left in the order forms a~clique. 

Let $LN_x \subset N_x$ be the set of all nodes adjacent to $x$ that lie on~the~left of~$x$
in $\pi$ and let $p_x$ be the right most node in $LN_x$. In the sequential version of the 
algorithm, for~each node $x$ we check if $LN_x - \{p_x\} \subset LN_{p_x}$. Now~we~do~this 
in parallel for~all nodes. 

The algorithm has two kernels. The first one, for each $x$, in parallel computes the left 
neighborhood $LN_x$ and the right most vertex in $LN_x$. In~the~second kernel, each thread 
$th_x$ processes the left neighborhood of $x$ and the left neighborhood of $p_x$. If some 
left neighbor of $x$, different from $p_x$, is not a left neighbor of $p_x$ then $th_x$ 
marks the global variable $flag$ on false. At~the~end, if $flag$ is true then the order 
is the perfect elimiantion order.\\

\begin{tabular}{l}
\hline
Parallel Test for Perfect Elimination Order\\
\hline
\\
//run on the cpu\\
parallelTestPEO()\\
\hspace{0.5cm}  $flag \leftarrow true$\\
\hspace{0.5cm}  preparationLNandP()\\
\hspace{0.5cm}  testing()\\
\hspace{0.5cm}  if $flag = true$ then return YES\\
\hspace{0.5cm}  else return NO\\
\\
//run on the gpu\\
preperationLNandP()\\
\hspace{0.5cm}  $p_x \leftarrow 0$\\
\hspace{0.5cm}  for each $y$ adjacent to $x$ do\\
\hspace{1cm}        if $order^{-1}(y) < order^{-1}(x)$ then\\
\hspace{1.5cm}          $LN_x$.insert(y)\\
\hspace{1.5cm}          if $order^{-1}(y) > order^{-1}(p_x)$ then\\
\hspace{2cm}                $p_x \leftarrow y$\\
\\
//run on the gpu\\
testing()\\
\hspace{0.5cm}  for each $y$ adjacent to $x$ do\\
\hspace{1cm}        if $y \in LN_x$ and $y \notin LN_{p_x}$ then\\
\hspace{1.5cm}          $flag \leftarrow false$\\
\\

\hline

\end{tabular}\\

\subsection{Details of the parallel implementation}

After~reading~the~input,~the~algorithm~copies~~the~~adjacency~~matrix of~the~graph to the array 
on the device memory. During the LexBFS algorithm, no other memory transfer between host 
and device is performed.

Since the algorithm does not compare any labels, the label concatenation can be omitted. 
Instead of this, the algorithm assignes to the new sets the~numbers which have not appeared 
yet. In each iteration, there are at~most $N$ new sets, hence during the whole algorithm 
there are at~most $N^2$ new sets. Since each set has a pointer to the next set on a list 
then~to~store all~pointers, the algorithm uses an array~of~~size~$N^2$,~~which~is~indexed 
by~the~numbers of the sets.

Our parallel implementation of the LexBFS algorithm uses the following arrays:

\begin{enumerate}
    \item the $Adj$ array of size $N^2$ for the boolean adjacency matrix.
    \item the $label$ array of size $N$ for the integer labels of sets.
    \item the $order$~~array~~of~~size~~$N$~~for~~the~~integer~~indices~~of~~vertices in~the~LexBFS~order.
    \item the $next\_label$ array of size $N^2$ for the integer indices of the next sets.
    \item the $old\_next\_label$ auxiliary array of size $N$ for the saved values 
        from~the~$next\_label$ array, one value for each vertex.
    \item the $counter$ array of size $N^2$ for the boolean flags for recognizing 
        if~a~set is empty.
    \item the $current$ variable for the integer number of vertex whose label is 
        lexicographically the largest.
\end{enumerate}

During an iteration, the algorithm processes one $current$ vertex which~is assigned to one row 
of every 2-dimensional array and all threads process that~row in one time. See figure below.

\begin{figure}[h]
    \begin{center}
        \includegraphics[width=5cm]{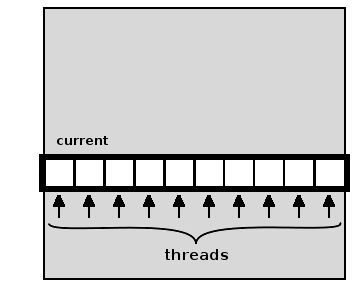}
    \end{center}
\end{figure}

Because each $current$ is unique and unrepeatable throughout the LexBFS algorithm then each row 
is visited only once. Therefore, in order to reduce the amount of the device memory used by our 
algorithm, we use the~\text{2-dimensional} $Adj$ array for two purposes: first as the adjacency matrix, 
next as the counter array. 

The second part of the chordality algorithm uses two arrays: the~$Adj$ array and the $order$ array. 
Because the~$Adj$ array is overwitten after the~LexBFS algorithm then the algorithm 
copies again the adjacency matrix from the host memory to the $Adj$ array on the device memory.

\section{Tests and results}

We introduce the following terminology to be used in this section. A~graph $G=(V, E)$ with the 
vertex set $V$ of size $N$ is $sparse$ if the size of $E$ is $\theta(N)$. A graph $G$ is $dense$ 
if the size of $E$ is $\theta(N^2)$. We consider the following classes of graphs:

\begin{enumerate}
    \item Cliques on $N$-vertices, for $N \in \{1000, 2000, 3000, \ldots, 10000\}$.
    \item Dense random graphs on $10000$ vertices.
    \item Sparse random graphs on $10000$ vertices.
    \item Trees on $10000$ vertices.
    \item Chordal random graphs on $10000$ vertices.
\end{enumerate}

We test two implementations. The sequential implementation is the~Habib, McConnell, Paul and 
Viennot algorithm presented in [2], which use a static memory allocation. For each class, we 
also present the time excluding the~input reading and the dynamic allocation of the device 
memory. For the parallel implementation, the time of reading the input and the dynamic allocation 
on the GPU is many times greater than the remaining time of the algorithm. However algorithm 
cannot be implemented without these operations.

\newpage
\subsection{Cliques}

Figure 6 presents timing results for cliques. For graphs of size smaller than 1000, the
sequential version is faster. When vertices number is 10000, the parallel implementation is 
two times faster than the sequential one. 

\begin{figure}[h!]
    \caption{Cliques}
    \begin{center}
    \begin{tabular}{|r|c|c|c|c|}
        \hline
        & \multicolumn{4}{|c|}{the time (ms)}\\
        \cline{2-5}
        & \multicolumn{2}{|c|}{GPU} & \multicolumn{2}{|c|}{CPU}\\
        \cline{2-5}
            & without input and &  & without & \\
        N   & memory allocation time &  & input time & \\
        \hline
        1000 & 0.4 &  1.5 &  1.3 &  2.1\\
        2000 & 0.8 &  4.9 &  4.9 &  8.3\\
        3000 & 1.4 &  9.8 & 11.1 & 18.9\\
        4000 & 2.1 & 17.0 & 19.5 & 33.7\\
        5000 & 2.7 & 26.2 & 30.5 & 52.1\\
        6000 & 3.7 & 37.1 & 44.0 & 74.8\\
        7000 & 4.4 & 50.1 & 60.0 &101.7\\
        8000 & 5.4 & 66.0 & 77.8 &132.6\\
        9000 & 6.6 & 83.2 & 99.0 &168.3\\
       10000 & 7.8 &101.8 &121.9 &207.3\\
       11000 & 8.9 &126.0 &147.0 &251.4\\
        \hline
    \end{tabular}
    \includegraphics[width=14.1cm]{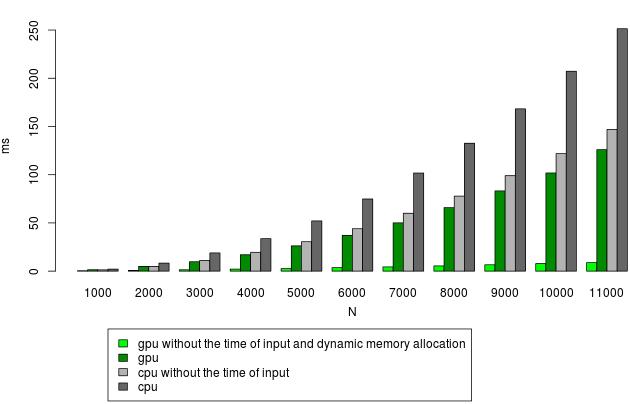}\\
    \end{center}
\end{figure}

\newpage
\subsection{Dense graphs}

Figure 7 presents timing results for dense random graphs. For each test the parallel 
implementation is almost two times faster than the sequential implementation.

\begin{figure}[h!]
    \caption{Dense random graphs: $N=10000$, $M = O(N^2)$}
    \begin{center}
    \begin{tabular}{|r|c|c|c|c|}
        \hline
        & \multicolumn{4}{|c|}{the time (ms)}\\
        \cline{2-5}
        & \multicolumn{2}{|c|}{GPU} & \multicolumn{2}{|c|}{CPU}\\
        \cline{2-5}
            & without input and &  & without & \\
           & memory allocation time &  & input time & \\
        \hline
        test1 & 9.0 & 107.4 & 107.1 & 191.9\\
        test2 & 8.9 & 108.6 & 106.6 & 191.0\\
        test3 & 8.9 & 106.2 & 107.4 & 191.6\\
        test4 & 8.9 & 106.7 & 106.4 & 191.3\\
        test5 & 8.9 & 107.2 & 109.0 & 191.3\\
        \hline
    \end{tabular}
    \includegraphics[width=13cm]{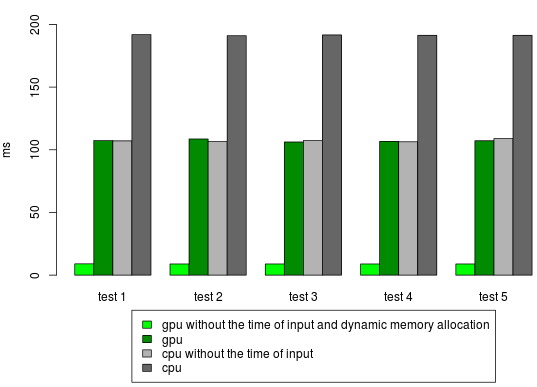}\\
    \end{center}
\end{figure}
\newpage

\subsection{Sparse graphs}

We have tested our implementation on sparse random graphs which $M=20N$.
The parallel implementation is slower than the sequential implementation.
(Figure 8)

\begin{figure}[h!]
    \caption{Sparse random graphs: $N=10000$, $M=20N$}
    \begin{center}
    \begin{tabular}{|r|c|c|c|c|}
        \hline
        & \multicolumn{4}{|c|}{the time (ms)}\\
        \cline{2-5}
        & \multicolumn{2}{|c|}{GPU} & \multicolumn{2}{|c|}{CPU}\\
        \cline{2-5}
            & without input and &  & without & \\
           & memory allocation time&  & input time & \\
        \hline
        test1 & 11.3 & 92.7 & 0.9 & 71.6\\
        test2 & 11.2 & 91.5 & 0.7 & 71.0\\
        test3 & 11.2 & 91.1 & 0.8 & 71.1\\
        test4 & 11.2 & 90.9 & 0.8 & 71.1\\
        test5 & 11.2 & 92.0 & 0.8 & 72.0\\
        \hline
    \end{tabular}

    \includegraphics[width=13cm]{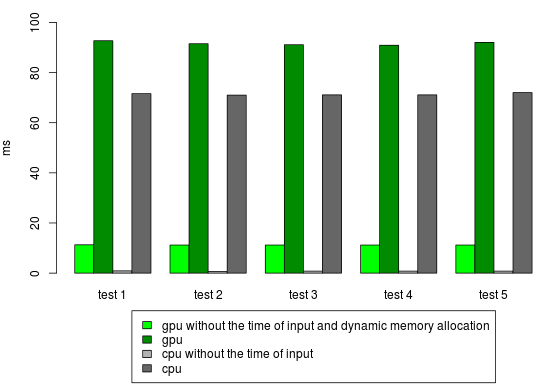}\\
    \end{center}
\end{figure}
\newpage

\subsection{Trees}

The results for trees are very similar to the results for sparse random graphs (Figure 9).

\begin{figure}[h!]
    \caption{Trees: N=10000}
    \begin{center}
    \begin{tabular}{|r|c|c|c|c|}
        \hline
        & \multicolumn{4}{|c|}{the time (ms)}\\
        \cline{2-5}
        & \multicolumn{2}{|c|}{GPU} & \multicolumn{2}{|c|}{CPU}\\
        \cline{2-5}
            & without input and &  & without & \\
           & memory allocation time&  & input time & \\
        \hline
        test1 & 7.2 & 86.9 & 1.4 & 73.0\\
        test2 & 7.3 & 87.6 & 1.4 & 73.3\\
        test3 & 7.4 & 87.6 & 1.1 & 71.0\\
        test4 & 7.3 & 86.6 & 0.0 & 70.3\\
        test5 & 7.3 & 88.0 & 0.1 & 70.0\\
        test6 & 7.5 & 87.7 & 0.1 & 70.3\\
        test7 & 7.2 & 87.6 & 1.0 & 71.2\\
        \hline
    \end{tabular}
        \includegraphics[width=13cm]{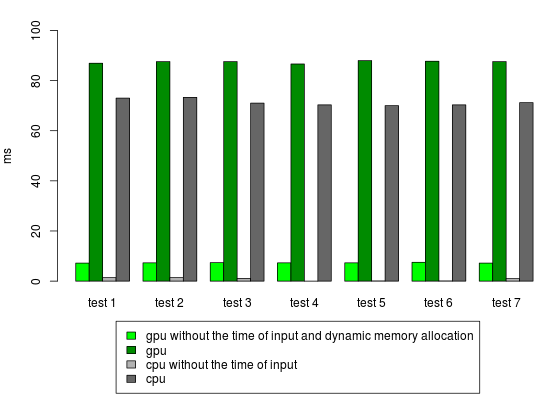}\\
    \end{center}
\end{figure}
\newpage

\subsection{Chordal graphs}

Figure 10 presents timing results for chordal random graphs, including dense and sparse 
graphs. Only for sparse graphs the parallel implementation is slower. On this figure it is easy to see
that the parallel implementation is stable - the time of algorithm is independent from the number of 
edges, in~contrast to the sequential implementation.

\begin{figure}[h!]
    \caption{Chordal random graphs, N=10000}
    \begin{center}
    \begin{tabular}{|r|c|c|c|c|}
        \hline
        & \multicolumn{4}{|c|}{the time (ms)}\\
        \cline{2-5}
        & \multicolumn{2}{|c|}{GPU} & \multicolumn{2}{|c|}{CPU}\\
        \cline{2-5}
            & without input and &  & without & \\
           & memory allocation time&  & input time & \\
        \hline
        test1 & 7.2 &  92.0 & 19.2 & 92.8\\
        test2 & 7.9 &  99.0 & 66.4 & 146.5\\
        test3 & 7.9 &  99.2 & 68.9 & 149.2\\
        test4 & 7.6 &  98.1 & 62.8 & 142.5\\
        test5 & 7.8 &  95.7 & 42.4 & 120.2\\
        test6 & 7.4 &  90.0 & 12.8 & 86.0\\
        test7 & 7.5 &  90.7 & 13.1 & 85.4\\
        test8 & 7.4 &  90.1 & 11.7 & 83.8\\
        \hline
    \end{tabular}
    \includegraphics[width=13cm]{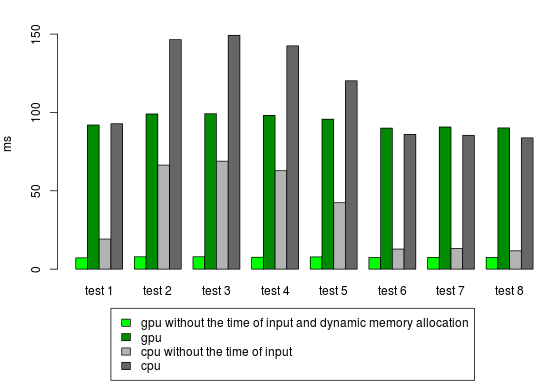}
    \end{center}
\end{figure}

\section{Conclusion and Future Work}

The main result of this paper is the parallel algorithm to test chordality of graphs 
based on our own efficient parallel version the LexBFS algorithm. For a graph $G$ of 
$N$ vertices and $M$ edges, the algorithm takes the $O(N)$ time and performs the $O(N^2)$ 
work on the $N$-threads machine. We use the~CUDA multithreads architecture to implement 
these algorithms.

Our parallel implementation achives best results for cliques and dense graphs. For 
graphs of 1000 and more vertices, the parallel algorithm is significantly faster 
than our fast sequential implementation and for graphs of~10000 vertices, the parallel 
implementation is two times faster than~the~sequential version. For trees, sparse graphs 
and small graphs (less than 1000 vertices) the sequential algorithm outperforms the
parallel one. However, for~this kind of data the parallel implementation is stable, 
the execution time is independent of the size of a graph.

It would be interesting if the parallel LexBFS algorithm could be used as~a~core for 
efficient parallel testing of interval graphs. Further research could be also made 
towards parallel implementation of the MCS algorithm.

\end{document}